\documentclass[letterpaper]{article}
\usepackage[utf8]{inputenc}
\usepackage[
  biblio=texdefault,
  linkoptions={capitalise,noabbrev},
  macros=false,
  boxes=false,
  algorithms=false,
  graphics=false
]{pomegranate}
\usepackage{needspace}
\usepackage{tikz}
\usepackage{bookmark}
\hypersetup{pdftitle={Sampling Line-Graph Colorings with Constant Extra Colors},
 pdfauthor={Alireza Haqi},
 pdfsubject={Bochner identities, Glauber dynamics, and spectral independence}}
\allowdisplaybreaks[1]
\numberwithin{equation}{section}

\AddToHook{env/corollary/before}{\Needspace{5\baselineskip}}
\DeclareMathOperator{\Var}{Var}

\DeclareMathOperator{\Ent}{Ent}
\DeclareMathOperator{\gap}{gap}
\newcommand{\E}{\mathbb E}
\newcommand{\R}{\mathbb R}
\newcommand{\Lop}{\mathcal L}
\newcommand{\cE}{\mathcal E}
\newcommand{\GL}{\mathrm{GL}}
\newcommand{\mix}{\mathrm{mix}}
\newcommand{\mLS}{\mathrm{mLS}}
\newcommand{\TV}{\mathrm{TV}}
\newcommand{\one}{\mathbf 1}
\newcommand{\ip}[2]{\langle #1,#2\rangle_\mu}
\newcommand{\norm}[1]{\lVert #1\rVert_\mu}

\title{Sampling Line-Graph Colorings with Constant Extra Colors}
\author[1]{Alireza Haqi}
\affil[1]{Stanford University, \href{mailto:ahaqi@stanford.edu}{\texttt{ahaqi@stanford.edu}}}
\date{}

\begin{document}
\maketitle

\begin{abstract}
Let $G$ be the line graph of a finite simple graph, with $n\geq1$ vertices
and maximum degree $\Delta$. We prove that single-site Glauber dynamics
for uniform proper $q$-colorings mixes in
$O_\Delta(n\log(n/\varepsilon))$ steps for every integer
$q\geq\Delta+5$.
Our proof uses the Bochner framework of Chen and Liu~\cite{CL26}.
\end{abstract}

\section{Introduction}
\label{sec:intro}

Sampling a proper coloring uniformly at random is a fundamental problem in
approximate counting and Markov chain algorithms. Its natural local
algorithm is the single-site Glauber chain: choose a vertex uniformly and
replace its color by a uniformly random feasible color, including its current
color. Each step inspects only one vertex and its neighborhood. The main
question is how many colors suffice for this local procedure to converge
quickly from every proper initial coloring.

The coloring-dynamics conjecture predicts mixing in time polynomial in
$n$ for $q\geq\Delta+2$, even when the maximum degree $\Delta$ grows with
$n$. We state its bounded-degree form below. For further background, see
the introduction of Wang, Zhang, and Zhang~\cite{WZZ}.
\begin{conjecture}[Coloring dynamics, bounded-degree form]
\label{conj:coloring}
For every fixed integer $\Delta\geq0$ and every fixed integer
$q\geq\Delta+2$, the Glauber chain on proper $q$-colorings of an
$n$-vertex graph of maximum degree at most $\Delta$ has mixing time
polynomial in $n$, at total variation error $1/4$.
\end{conjecture}
We study line graphs, which connect this question directly to sampling edge
colorings. For a finite simple graph $H$, the line graph $L(H)$ has vertex
set $E(H)$: distinct edges $e,f\in E(H)$ are adjacent as vertices of $L(H)$
exactly when they share an endpoint in $H$. Thus proper vertex colorings
of $L(H)$ are precisely
proper edge colorings of $H$. Throughout the paper, $\Delta$ denotes a
degree bound for the \emph{line graph}; the maximum degree of $H$ is
denoted by $d$ when needed.

\subsection{Results}

Let $P_\GL$ be the Glauber transition matrix. A \emph{feasible pinning}
fixes the colors on a subset of vertices to values that extend to a proper
coloring. If $\tau$ leaves $k$ vertices unpinned, $P_\GL^\tau$ denotes the
conditional Glauber chain, choosing uniformly among those $k$ vertices.
Our mixing-time convention counts individual updates.

\begin{theorem}[Main theorem]
\label{thm:main}
There is a sequence of positive numbers $(C_\Delta)_{\Delta\geq0}$ with the
following property. Let $H$ be a finite simple graph such that $G=L(H)$
has $n\geq1$ vertices and maximum degree $\Delta$. For every integer
$q\geq\Delta+5$, set $\kappa=q-\Delta-4$. Then, for every
$0<\varepsilon<1/2$,
\begin{equation}
 \begin{aligned}
 \gap(P_\GL)&\geq\frac{\kappa}{nq},\\
 t_\mix(P_\GL,\varepsilon)&\leq
 \min\left\{C_\Delta n\log\frac n\varepsilon,\;
 \left\lceil\frac{nq}{\kappa}
 \left(\frac n2\log q+\log\frac1{2\varepsilon}\right)\right\rceil\right\}.
 \end{aligned}
 \label{eq:main}
\end{equation}
For every feasible pinning $\tau$ with $k\geq1$ unpinned vertices,
\begin{equation}
 \gap(P_\GL^\tau)\geq\frac{\kappa}{kq},
 \qquad
 \lambda_{\max}(\Psi^\tau)\leq
 \eta:=\frac{\Delta+4}{q-\Delta-4},
 \label{eq:conditional-main}
\end{equation}
where $\Psi^\tau$ is the conditional influence matrix defined in
\cref{sec:si}. For each fixed pair $(\Delta,q)$ in this range, there
is also a constant $c_{\Delta,q}>0$, depending only on $\Delta$ and $q$,
such that
\begin{equation}
 \rho_\mLS(P_\GL)\geq\frac{c_{\Delta,q}}n.
 \label{eq:mlsi-main}
\end{equation}
\end{theorem}

The color threshold is three colors above the threshold in
\cref{conj:coloring}, restricted to line graphs. The second mixing bound
follows directly from the spectral gap; at $q=\Delta+5$, it is
$O(n^2(\Delta+5)\log(\Delta+5))$ for fixed error. The constant $C_\Delta$
in the first bound may depend on $\Delta$; it is uniform in the palette
size $q$. The modified log-Sobolev estimate stated here allows dependence
on both $\Delta$ and $q$.

The conditional estimates are proved in \cref{sec:gap,sec:si}, and the mixing and entropy estimates in
\cref{sec:mix}.

\subsection{Prior work}

For general graphs, Jerrum~\cite{Jerrum} established efficient approximate counting of
colorings above the $2\Delta$ threshold using a rapidly mixing coloring
chain. Carlson and Vigoda~\cite[Theorem~1.1 and Corollary~1.2]{CV}
proved $O(n\log n)$ mixing for flip dynamics when $\Delta\geq125$ and
$q\geq1.809\Delta$, and deduced $O_{\Delta,q}(n\log n)$ mixing
for Glauber dynamics in the same range.

For line graphs, Abdolazimi, Liu, and Oveis Gharan~\cite[Theorem~1.1]{ALG}
introduced the matrix trickle-down method and proved that, for every fixed
$\delta>0$ and all sufficiently large root-graph maximum degrees
$d\geq d_0(\delta)$, Glauber dynamics for proper edge colorings mixes in
$n^{O_\delta(1)}$ steps at total variation error $1/4$ whenever
$q\geq(10/3+\delta)d$, where $n$ is the number of edges of the root graph.
Their theorem also covers edge-list-colorings.
Building on their method, Wang, Zhang, and Zhang~\cite[Theorem~1]{WZZ}
proved $O_\Delta(n\log n)$ mixing on line graphs when
$q\geq\Delta+O(\Delta/\log\Delta)$ for $\Delta>1$.
Here $\Delta$ is the line-graph degree bound, whereas $d$ above is the
root-graph degree. \cref{thm:main} replaces the degree-dependent
surplus in the line-graph bound by the additive constant five.

In the special case of edge colorings of trees, smaller palettes suffice.
Delcourt, Heinrich, and Perarnau~\cite{DHP} proved polynomial mixing
for Glauber dynamics when $q\geq d+1$, where $d$ is the maximum degree
of the root tree. Carlson, Chen, Feng, and Vigoda~\cite[Theorems~1 and~3]{CCFV}
proved $O_{d,q}(n)$ relaxation time, the inverse spectral gap, for
$q\geq d+2$ on arbitrary trees, and $O_{d,q}(n\log^2 n)$ mixing time
on complete regular trees in this range. These results exploit the
tree structure, whereas \cref{thm:main} applies to line graphs of
arbitrary finite simple root graphs. The degree parameters differ:
a root graph of maximum degree $d$ has line-graph maximum degree at
most $2d-2$.

Our analytic starting point is the discrete Bochner method of Boudou,
Caputo, Dai Pra, and Posta~\cite{BCDP}. We give a direct finite-sum identity
adapted to uniform coloring measures. The line-graph structure enters only
through an unsigned incidence matrix. Spectral independence is then a
consequence of the resulting conditional Poincar\'e inequalities, rather
than the output of a matrix trickle-down argument. To pass from these
variance estimates to $n\log n$ mixing, we use the entropy theorem of Chen,
Liu, and Vigoda~\cite[Theorem~1.12]{CLV}.

Related work of Chen, Wang, Zhang, and Zhang~\cite{CWZZ} studies coupling
independence and spatial mixing for edge colorings, including coupling
independence for $q\geq3d$ on graphs of maximum degree $d$. Those
correlation-decay statements concern different properties from the
Glauber bounds proved here. In particular, our argument makes no claim
of strong spatial mixing or deterministic approximate counting.

Recent work develops integrated
Bochner identities and related squared-generator arguments for sampling.
G\"obel, Jenssen, Michelen, Pappik, Perkins, and Schiller~\cite{GJMPS}
give a short Bochner proof of rapid mixing for
the hard-core model on random regular graphs beyond the tree-uniqueness
threshold. Guo and Zhang~\cite{GZ} use a related squared-generator
criterion for hard-core sampling and approximate counting on planar graphs
at sufficiently small constant activity. Chen and Liu~\cite{CLG26}
localize the Bochner expansion to stars in graphs of girth at least five,
obtaining polynomial mixing for proper $q$-colorings when
$q\geq(1+\delta)D$, for every fixed $\delta\in(0,1)$ and sufficiently
large maximum degree $D\geq D_0(\delta)$. For the Sherrington--Kirkpatrick
model, Wang~\cite{Wang26} combines the integrated Bakry--\'Emery criterion
with localization to obtain $O_\beta(n\log n)$ Glauber mixing at error
$1/4$ for every fixed inverse temperature $0\leq\beta<1/2$, with high
probability over the disorder and uniformly over external fields.
Using a related rank-one-perturbed trickle-down argument, Boban, Li, and
Oveis Gharan~\cite{BLO} obtain $O(n^2)$ mixing at error $1/4$ for the
zero-field model when $0\leq\beta<1/2+\varepsilon_0$, for some absolute
$\varepsilon_0>0$, again with high probability over the disorder.
Chen and Liu~\cite{CL26} make the connection to trickle-down explicit,
using the integrated Bochner criterion to reprove and sharpen
spectral-gap bounds for down-up walks; \cref{sec:projection-framework}
explains the connection with our argument.

\subsection{Proof overview}

We analyze the chain that assigns rate one to each feasible recoloring.
In its Bochner identity, compatible pairs of color changes form squares
and contribute nonnegative terms. Incompatible pairs assign the same
color to adjacent vertices, producing a quadratic form in the adjacency
matrix $A_G$. For a line graph $G=L(H)$, the unsigned incidence matrix $J$
of $H$ satisfies $A_G+2I=J^{\mathsf T}J$. This bounds the quadratic form
below by minus four times the Dirichlet form. The diagonal terms contribute
at least $q-\Delta$ times the Dirichlet form, giving a spectral gap of at
least $q-\Delta-4$; see \cref{sec:bochner,sec:gap}.

The same argument applies after every feasible pinning. Comparing
Dirichlet forms transfers the bound to the Glauber chain, whose gap is at
least $(q-\Delta-4)/(kq)$ when $k\geq1$ vertices remain unpinned.
The universality theorem of Anari et al.~\cite{AJKPV} then gives spectral
independence; see \cref{sec:si}.

A lower bound on the one-vertex marginals allows us to apply the entropy
theorem of Chen, Liu, and Vigoda~\cite{CLV}, yielding the modified
log-Sobolev estimate for every $q\geq\Delta+5$. For
$\Delta+5\leq q\leq2\Delta$, the parameters of this theorem depend only
on $\Delta$. Jerrum's classical coupling~\cite{Jerrum} covers
$q\geq2\Delta+1$. Together, these give
$O_\Delta(n\log(n/\varepsilon))$ mixing uniformly in $q$, as shown in
\cref{sec:mix}.

\subsection{Acknowledgments}
The author is supported by NSF grant CCF-2045354.
The author acknowledges OpenAI's GPT-6 Astra for discovering the ideas
underlying this work. The manuscript was written and revised collaboratively
by the author and GPT-6 Astra.

\section{Preliminaries}
\label{sec:prelim}

\subsection{Coloring measures and Markov chain conventions}

We use $u,v$ for vertices of the graph $G$ whose vertices are colored.
This convention also applies when $G=L(H)$, so each vertex of $G$
represents an edge of $H$. In \cref{sec:incidence} only, we work directly
with edge colorings of $H$: $e,f\in E(H)$ index the colored sites, and
$u,v\in V(H)$ denote vertices of the root graph. From \cref{sec:gap}
onward we return to vertices of $G$. Scalar test functions are denoted by
$\varphi$ or $h$.

Let $G=(V,E)$ be a finite simple graph, and write $d_G(v)$ for the degree
of $v$. For an integer $q\geq1$, write
$[q]=\{1,\ldots,q\}$ and assign each vertex a list $L_v\subseteq[q]$.
Let $\Omega$ be the set of proper list-colorings of $G$ and, whenever
$\Omega\neq\varnothing$, let $\mu$ be uniform on $\Omega$.
Ordinary $q$-colorings are the case $L_v=[q]$ for all $v$.
For real functions on $\Omega$, use
$\ip \varphi h=\E_\mu[\varphi h]$ and $\norm \varphi^2=\ip \varphi \varphi$.

For $\sigma\in\Omega$, set
\begin{equation}
 A_v(\sigma)=L_v\setminus\{\sigma(u):u\sim v\},
 \qquad a_v(\sigma)=|A_v(\sigma)|.
 \label{eq:available}
\end{equation}
Thus $A_v(\sigma)$ is the set of feasible colors at $v$ when all other
vertices are pinned according to $\sigma$.
The current color belongs to $A_v(\sigma)$. Both $A_v$ and $a_v$ depend
only on the configuration outside $v$. The Glauber update operator $P_v$
averages uniformly over the feasible colors at $v$, with all other
vertices pinned. Thus $P_v\varphi$ and $\Var_v(\varphi)$ are the conditional
expectation and variance given the colors outside $v$.

For a feasible color $c$, let $\sigma^{v\leftarrow c}$ denote the recolored
configuration. Define
\begin{equation}
 g_{v,c}(\sigma)=
 \begin{cases}
 \varphi(\sigma^{v\leftarrow c})-\varphi(\sigma),&c\in A_v(\sigma),\\
 0,&c\notin A_v(\sigma),
 \end{cases}
 \qquad g_c=(g_{v,c})_{v\in V}.
 \label{eq:gradients}
\end{equation}
The constant-rate positive Laplacian and its Dirichlet form are
\begin{equation}
 \begin{split}
 \Lop=\sum_{v\in V}\Lop_v,\qquad
 \Lop_v\varphi(\sigma)=\sum_{c\in A_v(\sigma)}
       \bigl(\varphi(\sigma)-\varphi(\sigma^{v\leftarrow c})\bigr),\\
 \cE(\varphi,\varphi)=\ip \varphi{\Lop \varphi}.
 \end{split}
 \label{eq:laplacian}
\end{equation}
Every nontrivial
move and its reverse have rate one, so $\Lop$ is self-adjoint and positive
semidefinite in $L^2(\mu)$.

Since $\Lop_v=a_v(I-P_v)$, conditioning on the colors outside $v$
and then averaging gives
\begin{equation}
 \begin{split}
 \cE(\varphi,\varphi)&=\frac12\E_\mu\sum_{v,c}g_{v,c}^2
       =\sum_v\E_\mu[a_v\Var_v(\varphi)],\\
 \norm{\Lop_v\varphi}^2&=\E_\mu[a_v^2\Var_v(\varphi)].
 \end{split}
 \label{eq:one-site-identities}
\end{equation}
Indeed, fixing the colors outside $v$ leaves a conditional state space
of size $a$, the number of available colors at $v$. On this space, the
matrix of $\Lop_v$ is $aI-\one\one^{\mathsf T}$.

For $|V|=n\geq1$, the Glauber operator is
$P_\GL=n^{-1}\sum_v P_v$. For any reversible transition matrix $P$ with
stationary distribution $\mu$, write
\[
 \cE_P(\varphi,\varphi)=\langle \varphi,(I-P)\varphi\rangle_\mu,\qquad
 \gap(P)=\inf_{\Var_\mu(\varphi)>0}\frac{\cE_P(\varphi,\varphi)}{\Var_\mu(\varphi)}.
\]
Use the analogous variational definition for $\gap(\Lop)$ with numerator
$\cE(\varphi,\varphi)$. On a nontrivial connected coloring space this is the smallest
positive eigenvalue of $\Lop$.
Variational infima over an empty set are interpreted as $+\infty$.
The mixing time is
\[
 t_\mix(P,\varepsilon)=\min\left\{t\in\mathbb Z_{\geq0}:
 \max_{\sigma\in\Omega}\|P^t(\sigma,\cdot)-\mu\|_\TV
 \leq\varepsilon\right\}.
\]
All logarithms are natural. A fully pinned configuration has mixing time
zero; we do not assign it a one-site transition kernel.

\subsection{The integrated Bochner criterion}
\label{sec:bochner-framework}

The Bochner method bounds the spectral gap by comparing the squared
Laplacian with its Dirichlet form. For self-adjoint operators, write
$A\succeq B$ if $\ip \varphi{A\varphi}\geq\ip \varphi{B\varphi}$ for every $\varphi$.
The following finite-state criterion appears in
\cite[Proposition~1.1]{BCDP}; see also~\cite[Lemma~1.3]{CL26}.

\begin{lemma}
\label{lem:bochner-criterion}
Let $K$ be a self-adjoint positive semidefinite operator on a finite
$L^2(\mu)$ space with kernel consisting exactly of the constant functions.
For every $\kappa>0$, the following are equivalent:
\begin{equation}
 \gap(K)\geq\kappa
 \quad\Longleftrightarrow\quad
 K^2\succeq\kappa K
 \quad\Longleftrightarrow\quad
 \norm{K\varphi}^2\geq\kappa\ip \varphi{K\varphi}\quad\text{for all }\varphi.
 \label{eq:bochner-criterion}
\end{equation}
\end{lemma}
An integrated Bochner identity rewrites $\norm{K\varphi}^2$ as a sum of
nonnegative square terms and a remainder that can be compared with
$\ip \varphi{K\varphi}$. \cref{sec:bochner} implements this
decomposition for recolorings and retains a signed remainder for the
incompatible pairs. Connectivity is checked separately in
\cref{sec:gap} to ensure the kernel condition in
\cref{lem:bochner-criterion}.

\subsection[The projection framework of Chen and Liu]{The projection framework of Chen and Liu~\cite{CL26}}
\label{sec:projection-framework}

Chen and Liu~\cite[Section~3]{CL26} apply the same criterion to sums of
conditional-expectation projections. In our notation, let
$D_v=I-P_v$ and $K=\sum_vD_v=n(I-P_\GL)$. Since $D_v^2=D_v$,
\begin{equation}
 \norm{K\varphi}^2=\sum_v\norm{D_v\varphi}^2
       +\sum_{u\ne v}\ip{D_u\varphi}{D_v\varphi},
 \qquad
 \ip \varphi{K\varphi}=\sum_v\norm{D_v\varphi}^2.
 \label{eq:projection-expansion}
\end{equation}
Suppose a symmetric matrix $C$ with nonnegative entries and $C_{vv}=0$
satisfies
\[
 \ip{D_u\varphi}{D_v\varphi}\geq-C_{uv}\norm{D_u\varphi}\norm{D_v\varphi}
 \qquad(u\ne v)
\]
for every $\varphi$. If $\lambda_{\max}(C)\leq1-\epsilon$ for some
$0<\epsilon<1$, then, with $x_v=\norm{D_v\varphi}$,
\[
 \norm{K\varphi}^2\geq x^{\mathsf T}(I-C)x
 \geq\epsilon\sum_vx_v^2=\epsilon\ip \varphi{K\varphi}.
\]
For an irreducible Glauber chain, \cref{lem:bochner-criterion}
therefore gives $\gap(P_\GL)\geq\epsilon/n$.
Chen and Liu derive these cross-term bounds from two-coordinate
conditional walks, yielding their spectral-influence criterion
\cite[Theorem~1.2 and Corollary~3.2]{CL26}.

Our identity uses $\Lop_v=a_vD_v$, whose rates are constant on feasible
moves. Grouping compatible recolorings into squares leaves a signed
quadratic form in $A_G$. Its lower bound gives
$\Lop^2\succeq(m-4)\Lop$ whenever $a_v\geq m$ on a line graph.
The Glauber gap follows by Dirichlet-form comparison in
\cref{eq:glauber-comparison}.

\section{An exact Bochner identity for recolorings}
\label{sec:bochner}

Using the Bochner method~\cite{BCDP}, we decompose
$\norm{\Lop\varphi}^2$ into nonnegative terms and a quadratic form in $A_G$.
The identity holds for arbitrary lists on any finite simple graph.
For line graphs, the root graph's incidence matrix bounds the quadratic
form from below, giving the spectral-gap estimate.

\subsection{Squares and incompatible moves}

Fix two distinct vertices $u,v$ and a feasible pinning $\tau$ on
$V\setminus\{u,v\}$. Choose a proper coloring $\sigma_{00}\in\Omega$
that agrees with $\tau$ on the pinned vertices, and write
$a=\sigma_{00}(u)$ and $b=\sigma_{00}(v)$.
Consider changing $u$'s color to $c\neq a$ and $v$'s color to
$d\neq b$, preserving the pinning. The four configurations have the
following colors at $(u,v)$:
\[
 \sigma_{00}: (a,b),\qquad
 \sigma_{10}: (c,b),\qquad
 \sigma_{01}: (a,d),\qquad
 \sigma_{11}: (c,d).
\]
The first index indicates whether $u$ has been recolored, and the second
whether $v$ has been recolored. If all four configurations belong to
$\Omega$, their set is a \emph{recoloring square}; see \cref{fig:squares}.

Let $\mathscr S$ be the collection of these squares. Each set of four
configurations is counted once, regardless of the starting corner or the
order of $u,v$. For each square, write
$\varphi_{ij}=\varphi(\sigma_{ij})$ and define
\begin{equation}
 \mathcal H(\varphi)=\frac{2}{|\Omega|}\sum_{S\in\mathscr S}
 (\varphi_{00}-\varphi_{10}-\varphi_{01}+\varphi_{11})^2.
 \label{eq:hessian}
\end{equation}
The alternating sum compares the change in $\varphi$ from recoloring $u$
before and after recoloring $v$. Its square does not depend on the starting
corner or the order of $u,v$.

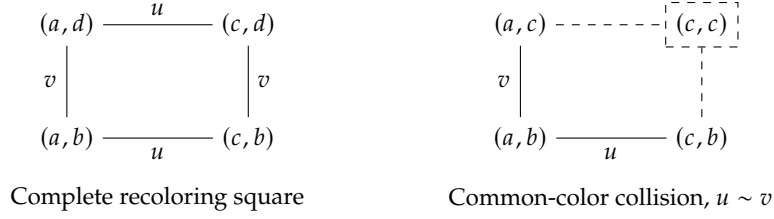
\begin{figure}[!ht]
\centering
\begin{tikzpicture}[x=1cm,y=1cm,every node/.style={font=\small}]
 \node (a) at (0,0) {$(a,b)$};
 \node (b) at (2.4,0) {$(c,b)$};
 \node (c) at (0,1.5) {$(a,d)$};
 \node (d) at (2.4,1.5) {$(c,d)$};
 \draw (a)--node[below] {$u$} (b);
 \draw (a)--node[left] {$v$} (c);
 \draw (b)--node[right] {$v$} (d);
 \draw (c)--node[above] {$u$} (d);
 \node at (1.2,-.8) {Complete recoloring square};
 \node (e) at (6,0) {$(a,b)$};
 \node (f) at (8.4,0) {$(c,b)$};
 \node (g) at (6,1.5) {$(a,c)$};
 \node[draw,dashed,inner sep=4pt] (h) at (8.4,1.5) {$(c,c)$};
 \draw (e)--node[below] {$u$} (f);
 \draw (e)--node[left] {$v$} (g);
 \draw[dashed] (f)--(h);
 \draw[dashed] (g)--(h);
 \node at (7.2,-.8) {Common-color collision, $u\sim v$};
\end{tikzpicture}
\caption{The two types of pairs of individually feasible, nontrivial moves,
shown schematically. Labels give the colors at $(u,v)$; all configurations
share the same pinning outside $\{u,v\}$. At right, the fourth configuration
is unfeasible.}
\label{fig:squares}
\end{figure}

\begin{lemma}
\label{lem:bochner}
For every finite simple graph $G$, every nonempty proper list-coloring
space $\Omega$, and every function $\varphi:\Omega\to\R$,
\begin{equation}
 \norm{\Lop \varphi}^2
 =\sum_v\E_\mu[a_v^2\Var_v(\varphi)]
  +\mathcal H(\varphi)+\E_\mu\sum_{c=1}^q g_c^{\mathsf T}A_Gg_c.
 \label{eq:bochner}
\end{equation}
\end{lemma}
\begin{proof}
Expanding the squared sum in~\cref{eq:laplacian} gives
\[
 \norm{\Lop \varphi}^2=\sum_v\norm{\Lop_v\varphi}^2
  +\E_\mu\sum_{u\neq v}\sum_{c,d}g_{u,c}g_{v,d}.
\]
The diagonal terms are given by~\cref{eq:one-site-identities}. A nonzero cross term
comes from individually feasible, nontrivial moves at distinct vertices.
If the joint recoloring is proper, its four corners form a unique square.
All four have the same stationary weight, and the reverse moves at the
other corners are feasible. For one ordering of $u,v$, the four
contributions are labeled below by their starting colorings:
\begin{align*}
 &\underbrace{(\varphi_{10}-\varphi_{00})(\varphi_{01}-\varphi_{00})}_{\sigma_{00}}
 +\underbrace{(\varphi_{00}-\varphi_{10})(\varphi_{11}-\varphi_{10})}_{\sigma_{10}}\\
 &\quad +\underbrace{(\varphi_{11}-\varphi_{01})(\varphi_{00}-\varphi_{01})}_{\sigma_{01}}
 +\underbrace{(\varphi_{01}-\varphi_{11})(\varphi_{10}-\varphi_{11})}_{\sigma_{11}}\\
 &= (\varphi_{00}-\varphi_{10}-\varphi_{01}+\varphi_{11})^2.
\end{align*}
There are two orderings of $u,v$, giving the factor two in
\cref{eq:hessian}.

If the joint recoloring is improper, every edge other than $uv$ remains
proper because both individual moves are feasible. Thus the only possible
obstruction is $u\sim v$ and $c=d$. Conversely, these conditions make the
joint recoloring improper. The remaining cross terms therefore have
$u\sim v$ and $c=d$. Since $g_{u,c}$ vanishes for unfeasible or trivial
moves, we may sum over every $c\in[q]$. With adjacency entries
$(A_G)_{u,v}=\one_{\{u\sim v\}}$ and the column vector
$g_c=(g_{v,c})_{v\in V}$, their contribution is
\[
 \begin{aligned}
 \E_\mu\sum_{c=1}^q\sum_{\substack{(u,v)\in V^2\\u\sim v}}g_{u,c}g_{v,c}
 &=\E_\mu\sum_{c=1}^q\sum_{u,v\in V}(A_G)_{u,v}g_{u,c}g_{v,c}\\
 &=\E_\mu\sum_{c=1}^q g_c^{\mathsf T}A_Gg_c.
 \end{aligned}
\]
The matrix identity holds at each configuration $\sigma\in\Omega$,
with all gradients evaluated at that same configuration, before taking
expectation. Each undirected edge contributes both ordered pairs
$(u,v)$ and $(v,u)$ on both sides, matching the sum over $u\neq v$ in
the initial expansion. This partitions all cross terms.
\end{proof}

\subsection{The incidence certificate for line graphs}
\label{sec:incidence}

Let $G=L(H)$, where $H$ is simple. In this subsection, we index the
colored sites by edges $e,f\in E(H)$ and reserve $u,v$ for vertices of
$H$. Thus $L_e$, $A_e$, $a_e$, $P_e$, $\Var_e$, and
$\Lop_e$ denote the lists, available colors, and one-site quantities
already defined for vertices of $G$; recoloring edge $e$ gives
$\sigma^{e\leftarrow c}$. In particular,
$g_{e,c}(\sigma)=\varphi(\sigma^{e\leftarrow c})-\varphi(\sigma)$ for
$c\in A_e(\sigma)$, and $g_{e,c}(\sigma)=0$ otherwise, so
$g_c=(g_{e,c})_{e\in E(H)}$.
Let $J$ be the unsigned incidence matrix of $H$, with rows indexed by
vertices $v\in V(H)$ and columns by edges $e\in E(H)$:
$J_{v,e}=\mathbf1_{\{v\in e\}}$. For edges $e,f$,
$(J^{\mathsf T}J)_{e,f}=|e\cap f|$. Every edge has two endpoints, and
distinct edges share at most one endpoint. Hence
\begin{equation}
 J^{\mathsf T}J=2I+A_G.
 \label{eq:incidence}
\end{equation}
The collision term in \cref{lem:bochner} therefore equals
\begin{equation}
 \E_\mu\sum_c g_c^{\mathsf T}A_Gg_c
 =\E_\mu\sum_c\|Jg_c\|_2^2-4\cE(\varphi,\varphi).
 \label{eq:collision}
\end{equation}
Here $(Jg_c)_v=\sum_{e\ni v}g_{e,c}$ for each vertex $v\in V(H)$.
The factor four is the product of the adjacency lower bound $-2I$ and
the factor two in $\E_\mu\sum_{e,c}g_{e,c}^2=2\cE(\varphi,\varphi)$.

\Needspace{9\baselineskip}
\begin{theorem}[Bochner certificate]
\label{thm:certificate}
Let $H$ be a finite simple graph, let $G=L(H)$, and let $\Omega$ be a
nonempty proper edge-list-coloring space of $H$. If a real number $m$
satisfies $a_e(\sigma)\geq m$ for every $e\in E(H)$ and
$\sigma\in\Omega$, then every
$\varphi:\Omega\to\R$ satisfies
\begin{equation}
 \begin{split}
 \norm{\Lop \varphi}^2-(m-4)\cE(\varphi,\varphi)
 ={}&\mathcal H(\varphi)+\E_\mu\sum_c\|Jg_c\|_2^2\\
 &+\sum_{e\in E(H)}\E_\mu[a_e(a_e-m)\Var_e(\varphi)]\geq0.
 \end{split}
 \label{eq:certificate}
\end{equation}
\end{theorem}
\begin{proof}
Substitute~\cref{eq:collision} in~\cref{eq:bochner}, then subtract
$(m-4)\cE(\varphi,\varphi)$ and use~\cref{eq:one-site-identities}. Each term on the right is
nonnegative.
\end{proof}

\section{Spectral gaps}
\label{sec:gap}

We return to vertex notation on $G=L(H)$: $u,v$ denote vertices of $G$,
each corresponding to an edge of $H$. In particular, $a_v$ denotes the
same feasible-color count written $a_e$ in \cref{sec:incidence}.

To convert~\cref{eq:certificate} into a Poincar\'e inequality for the
whole measure, we need connectivity of the coloring space. The following
standard fact also applies to the residual lists after pinning.

\begin{lemma}[Folklore]
\label{lem:connectivity}
If $|L_v|\geq d_G(v)+2$ for every vertex of a finite simple graph, then
proper list-colorings exist and any two are joined by a sequence of
feasible single-vertex recolorings.
\end{lemma}

\begin{corollary}
\label{cor:gap}
Suppose $G=L(H)$ for a finite simple graph $H$, the list-coloring space
is nonempty and connected, and $a_v(\sigma)\geq m>4$ everywhere. Then
\begin{equation}
 \gap(\Lop)\geq m-4,
 \qquad (m-4)\Var_\mu(\varphi)\leq\cE(\varphi,\varphi).
 \label{eq:poincare}
\end{equation}
\end{corollary}
\begin{proof}
Connectivity makes the kernel of $\Lop$ consist of the constant functions.
Apply \cref{lem:bochner-criterion} to
\cref{thm:certificate} with $\kappa=m-4$.
\end{proof}

Now suppose $G$ has maximum degree at most $\Delta$ and
$q\geq\Delta+5$. Put $m=q-\Delta$ and $\kappa=m-4$.
For a feasible pinning $\tau$ on $\Lambda\subseteq V$, write
$U=V\setminus\Lambda$, $F=G[U]$, and
\begin{equation}
 L_v^\tau=[q]\setminus\{\tau(u):u\in\Lambda,\ u\sim v\},
 \qquad v\in U.
 \label{eq:residual-lists}
\end{equation}
The conditional measure $\mu^\tau$ is uniform on proper list-colorings
of $F$. We use $a_v$ and $\Var_v$ for this conditional model and denote
its constant-rate Dirichlet form by $\cE_\tau$.
Since at most $d_G(v)-d_F(v)$ colors were removed,
\begin{equation}
 |L_v^\tau|\geq q-d_G(v)+d_F(v)\geq m+d_F(v).
 \label{eq:list-surplus}
\end{equation}
Consequently every remaining vertex has at least $m$ available colors,
and \cref{lem:connectivity} gives connectivity.
Moreover, $F$ is the line graph of the simple graph obtained from $H$
by deleting the edges corresponding to $\Lambda$. Thus
\cref{cor:gap} gives, uniformly in $\tau$,
\begin{equation}
 \kappa\Var_{\mu^\tau}(\varphi)\leq\cE_\tau(\varphi,\varphi).
 \label{eq:conditional-poincare}
\end{equation}

For $k=|U|\geq1$, compare with the conditional Glauber chain:
\begin{equation}
 \begin{split}
 \cE_\tau(\varphi,\varphi)
 &=\sum_{v\in U}\E_{\mu^\tau}[a_v\Var_v(\varphi)]\\
 &\leq q\sum_{v\in U}\E_{\mu^\tau}[\Var_v(\varphi)]
 =kq\,\cE_{P_\GL^\tau}(\varphi,\varphi).
 \end{split}
 \label{eq:glauber-comparison}
\end{equation}
Together with~\cref{eq:conditional-poincare}, this proves both gap
statements in \cref{thm:main}. Residual graphs may be disconnected;
the connectivity argument concerns their coloring spaces and still applies.

\section{Spectral independence}
\label{sec:si}

For a feasible pinning $\tau$ with $k\geq1$ unpinned vertices, let $\Psi^\tau$
be the influence matrix indexed by vertex-color pairs of positive marginal
probability. Following~\cite[Definitions~1.10--1.11]{CLV}, its entry at
$(u,c),(v,d)$ is
$\mu^\tau(\sigma_v=d\mid\sigma_u=c)-\mu^\tau(\sigma_v=d)$ for $u\neq v$,
and its within-vertex blocks are zero; spectral independence bounds
$\lambda_{\max}(\Psi^\tau)$ under every pinning. The universality theorem
of Anari, Jain, Koehler, Pham, and Vuong~\cite[Theorem~30]{AJKPV} implies
that a Glauber gap of at least $1/(Ck)$ gives
$\lambda_{\max}(\Psi^\tau)\leq C-1$ for $C\geq1$ in this convention:
on arrays centered with respect to each one-site marginal, their
correlation matrix acts as $I+\Psi^\tau$. Applying this theorem with $C=q/\kappa$ to the
gap proved in \cref{sec:gap} gives the following corollary.

\begin{corollary}
\label{cor:si}
Under the assumptions of \cref{thm:main}, every conditional
influence matrix satisfies
\begin{equation}
 \lambda_{\max}(\Psi^\tau)\leq\frac q\kappa-1
 =\frac{\Delta+4}{q-\Delta-4}.
 \label{eq:si}
\end{equation}
\end{corollary}

\section{Entropy and mixing with a uniform palette bound}
\label{sec:mix}

To obtain $O(n\log n)$ mixing for bounded degree, we apply an entropy
theorem after checking its remaining hypotheses.
For a reversible chain, our normalization is
\begin{equation}
 \rho_\mLS(P)=\inf_{\substack{\varphi>0\\\Ent_\mu(\varphi)>0}}
 \frac{\langle \varphi,(I-P)\log \varphi\rangle_\mu}{\Ent_\mu(\varphi)},
 \qquad
 \Ent_\mu(\varphi)=\E_\mu\left[\varphi\log\frac \varphi{\E_\mu \varphi}\right].
 \label{eq:mlsi}
\end{equation}

\begin{theorem}[Chen--Liu--Vigoda~\cite{CLV}, Theorem~1.12]
\label{thm:entropy}
Let $D\geq3$ be an integer and $b,\eta>0$. Consider a Gibbs distribution
on a graph with $n\geq1$ vertices and maximum degree at most $D$.
Suppose every
conditional configuration space is connected under single-site changes,
every positive one-site marginal under every pinning is at least $b$,
and every conditional influence matrix has largest eigenvalue at most
$\eta$. There are positive constants $c(D,b,\eta)$ and $C(D,b,\eta)$ such
that its Glauber chain satisfies
\[
 \rho_\mLS(P_\GL)\geq\frac{c(D,b,\eta)}n,
 \qquad
 t_\mix(P_\GL,\varepsilon)\leq
 C(D,b,\eta)n\log\frac n\varepsilon
 \quad(0<\varepsilon<1/2).
\]
\end{theorem}

\begin{lemma}
\label{lem:marginal}
Under the assumptions of \cref{thm:main}, for every feasible
pinning $\tau$, every unpinned vertex $v$, and every $c\in L_v^\tau$,
\begin{equation}
 \mu^\tau(\sigma_v=c)\geq
 b:=\frac1q\left(1-\frac1m\right)^\Delta,
 \qquad m=q-\Delta\geq5.
 \label{eq:marginal}
\end{equation}
\end{lemma}
\begin{proof}
Conditional on all colors other than that at a remaining vertex $u$,
any specified color has probability either zero or $1/a_u\leq1/m$.
List the unpinned neighbors of $v$ as $u_1,\ldots,u_\ell$ and let $B_j$
be the event that none of $u_1,\ldots,u_j$ uses $c$. The event $B_{j-1}$
is measurable without the color at $u_j$, so
\[
 \mu^\tau(B_j)
 =\E_{\mu^\tau}\!\left[
 \one_{B_{j-1}}\mu^\tau(\sigma_{u_j}\neq c\mid\sigma_{U\setminus\{u_j\}})
 \right]
 \geq(1-1/m)\mu^\tau(B_{j-1}).
\]
It follows that $\mu^\tau(B_\ell)\geq(1-1/m)^\ell$.
On $B_\ell$, the color $c$ is feasible at $v$, including compatibility with
the pinned neighbors by the definition of $L_v^\tau$. Its conditional
probability is then $1/a_v\geq1/q$. Averaging and using $\ell\leq\Delta$
proves~\cref{eq:marginal}.
\end{proof}

\begin{proof}[Proof of \cref{thm:main}]
The spectral-gap and influence bounds were proved in
\cref{sec:gap,sec:si}. Each $P_v$ is an orthogonal projection in
$L^2(\mu)$, so $P_\GL=n^{-1}\sum_vP_v$ has nonnegative spectrum.
Since $\mu$ is uniform and $|\Omega|\leq q^n$, the standard spectral
mixing estimate (see~\cite[Lemma~4]{WZZ}) gives, for every integer $t\geq0$,
\[
 \max_{\sigma\in\Omega}
 \|P_\GL^t(\sigma,\cdot)-\mu\|_\TV
 \leq\frac12\sqrt{|\Omega|}\,
 e^{-t\,\gap(P_\GL)}
 \leq\frac12 q^{n/2}e^{-t\kappa/(nq)}.
\]
Taking $t$ to be the ceiling in~\cref{eq:main} makes this at most
$\varepsilon$, proving the second mixing bound.

Proper coloring is a Gibbs model with interaction
$\one_{\{c\neq d\}}$ and unit vertex weights. Conditional connectivity
follows from~\cref{eq:list-surplus} and \cref{lem:connectivity}.
\cref{cor:si} and \cref{lem:marginal} verify the other
hypotheses of \cref{thm:entropy}, with
$D=\max\{3,\Delta\}$. This proves~\cref{eq:mlsi-main} and gives
$O_{\Delta,q}(n\log(n/\varepsilon))$ mixing.

To make the mixing constant independent of $q$, first consider
$\Delta+5\leq q\leq2\Delta$. In this range,
\begin{equation}
 \eta\leq\Delta+4,
 \qquad b\geq\frac1{2\Delta}\left(\frac45\right)^\Delta.
 \label{eq:bounded-palette}
\end{equation}
If this interval is nonempty then $\Delta\geq5$, and
\cref{thm:entropy} has parameters depending only on $\Delta$.
For $q\geq2\Delta+1$, the classical coupling argument of
Jerrum~\cite{Jerrum} gives $O_\Delta(n\log(n/\varepsilon))$ mixing,
uniformly in $q$.
\end{proof}

\begin{remark}[Dependence on $\Delta$]
The constant $C_\Delta$ in \cref{thm:main} can be chosen with
$C_\Delta\leq\exp(\exp(O(\Delta)))$. This doubly exponential upper
bound follows by substituting~\cref{eq:bounded-palette} into the
quantitative estimate of Chen--Liu--Vigoda~\cite[Theorem~1.12]{CLV}.
It describes the bound supplied by our proof; the optimal constant may
have a smaller dependence on $\Delta$.
\end{remark}

\section{Limitations and further questions}
\label{sec:questions}

The factor four in~\cref{eq:certificate} gives only a zero lower bound
when $q=\Delta+4$. This is a limitation of this estimate, not a
counterexample to rapid mixing. We leave determining the mixing time for
$q=\Delta+2$ as an open problem. We also leave open the optimal
dependence of $C_\Delta$ on $\Delta$ in the mixing bound
of~\cref{thm:main}.

\Needspace{10\baselineskip}


\begin{thebibliography}{BCDP06}

\bibitem[ALG21]{ALG}
Dorna Abdolazimi, Kuikui Liu, and Shayan Oveis Gharan.
\newblock A matrix trickle-down theorem on simplicial complexes and
applications to sampling colorings.
\newblock In \emph{62nd IEEE Symposium on Foundations of Computer Science
(FOCS)}, 2021.
\newblock \href{https://doi.org/10.1109/FOCS52979.2021.00024}{doi:10.1109/FOCS52979.2021.00024}.
\newblock Full version: \href{https://arxiv.org/abs/2106.03845}{arXiv:2106.03845}.

\bibitem[AJK$^{+}$24]{AJKPV}
Nima Anari, Vishesh Jain, Frederic Koehler, Huy Tuan Pham, and
Thuy-Duong Vuong.
\newblock Universality of spectral independence with applications to
fast mixing in spin glasses.
\newblock In \emph{Proceedings of the 2024 Annual ACM-SIAM Symposium on
Discrete Algorithms (SODA)}, pages 5029--5056, 2024.
\newblock \href{https://doi.org/10.1137/1.9781611977912.181}{doi:10.1137/1.9781611977912.181}.
\newblock Full version: \href{https://arxiv.org/abs/2307.10466v1}{arXiv:2307.10466v1}.

\bibitem[BLO26]{BLO}
Mathews Boban, Anqi Li, and Shayan Oveis Gharan.
\newblock Rank-1-perturbed trickledown theorems: Mixing time of Glauber
dynamics for the Sherrington--Kirkpatrick model up to
$\beta\leq1/2+\varepsilon$.
\newblock Preprint, 2026.
\newblock \href{https://arxiv.org/abs/2609.13138v1}{arXiv:2609.13138v1}.

\bibitem[BCDP06]{BCDP}
Anne-S\'everine Boudou, Pietro Caputo, Paolo Dai Pra, and Gustavo Posta.
\newblock Spectral gap estimates for interacting particle systems via a
Bochner-type identity.
\newblock \emph{Journal of Functional Analysis}, 232(1):222--258, 2006.
\newblock \href{https://arxiv.org/abs/math/0505533v3}{arXiv:math/0505533v3}.

\bibitem[CCFV25]{CCFV}
Charlie Carlson, Xiaoyu Chen, Weiming Feng, and Eric Vigoda.
\newblock Optimal mixing for randomly sampling edge colorings on trees
down to the max degree.
\newblock In \emph{Proceedings of the 2025 Annual ACM-SIAM Symposium on
Discrete Algorithms (SODA)}, pages 5418--5433, 2025.
\newblock \href{https://doi.org/10.1137/1.9781611978322.184}{doi:10.1137/1.9781611978322.184}.
\newblock Full version: \href{https://arxiv.org/abs/2407.04576v1}{arXiv:2407.04576v1}.

\bibitem[CV25]{CV}
Charlie Carlson and Eric Vigoda.
\newblock Flip dynamics for sampling colorings: Improving
$(11/6-\varepsilon)$ using a simple metric.
\newblock In \emph{Proceedings of the 2025 Annual ACM-SIAM Symposium on
Discrete Algorithms (SODA)}, pages 2194--2212, 2025.
\newblock \href{https://doi.org/10.1137/1.9781611978322.71}{doi:10.1137/1.9781611978322.71}.
\newblock Full version: \href{https://arxiv.org/abs/2407.04870v2}{arXiv:2407.04870v2}.

\bibitem[CL26a]{CLG26}
Xiaoyu Chen and Kuikui Liu.
\newblock A spectral local-to-global principle for spin systems on graphs
with girth at least five.
\newblock Preprint, 2026.
\newblock \href{https://arxiv.org/abs/2608.25491v1}{arXiv:2608.25491v1}.

\bibitem[CL26b]{CL26}
Xiaoyu Chen and Kuikui Liu.
\newblock Spectral gap of down-up walks via trickle-down: A simplified
and sharpened analysis.
\newblock Preprint, 2026.
\newblock \href{https://arxiv.org/abs/2609.19514v1}{arXiv:2609.19514v1}.

\bibitem[CLV26]{CLV}
Zongchen Chen, Kuikui Liu, and Eric Vigoda.
\newblock Optimal mixing of Glauber dynamics: Entropy factorization via
high-dimensional expansion.
\newblock \emph{SIAM Journal on Computing},
55(4):STOC21-104--STOC21-153, 2026. Published online in 2023.
\newblock \href{https://doi.org/10.1137/21M1443340}{doi:10.1137/21M1443340}.
\newblock Full version: \href{https://arxiv.org/abs/2011.02075v4}{arXiv:2011.02075v4}.

\bibitem[CWZZ25]{CWZZ}
Zejia Chen, Yulin Wang, Chihao Zhang, and Zihan Zhang.
\newblock Decay of correlation for edge colorings when $q>3\Delta$.
\newblock In \emph{52nd International Colloquium on Automata, Languages,
and Programming (ICALP)}, LIPIcs 334, Article 54, 2025.
\newblock \href{https://doi.org/10.4230/LIPIcs.ICALP.2025.54}{doi:10.4230/LIPIcs.ICALP.2025.54}.

\bibitem[DHP20]{DHP}
Michelle Delcourt, Marc Heinrich, and Guillem Perarnau.
\newblock The Glauber dynamics for edge-colorings of trees.
\newblock \emph{Random Structures \& Algorithms}, 57(4):1050--1076, 2020.
\newblock \href{https://doi.org/10.1002/rsa.20960}{doi:10.1002/rsa.20960}.
\newblock Full version: \href{https://arxiv.org/abs/1812.05577}{arXiv:1812.05577}.

\bibitem[GJM$^{+}$26]{GJMPS}
Andreas G\"obel, Matthew Jenssen, Marcus Michelen, Marcus Pappik,
Will Perkins, and Leon Schiller.
\newblock A simple proof of rapid mixing on random regular graphs beyond
uniqueness.
\newblock Preprint, 2026.
\newblock \href{https://arxiv.org/abs/2606.27545v1}{arXiv:2606.27545v1}.

\bibitem[GZ26]{GZ}
Heng Guo and Xinyuan Zhang.
\newblock Approximating spin systems on planar graphs.
\newblock Preprint, 2026.
\newblock \href{https://arxiv.org/abs/2608.06172v2}{arXiv:2608.06172v2}.

\bibitem[Jer95]{Jerrum}
Mark Jerrum.
\newblock A very simple algorithm for estimating the number of $k$-colorings
of a low-degree graph.
\newblock \emph{Random Structures \& Algorithms}, 7(2):157--165, 1995.
\newblock \href{https://doi.org/10.1002/rsa.3240070205}{doi:10.1002/rsa.3240070205}.

\bibitem[Wan26]{Wang26}
Sihan Wang.
\newblock Optimal mixing of Glauber dynamics for the
Sherrington--Kirkpatrick model at $\beta<1/2$.
\newblock Preprint, 2026.
\newblock \href{https://arxiv.org/abs/2608.22159v2}{arXiv:2608.22159v2}.

\bibitem[WZZ24]{WZZ}
Yulin Wang, Chihao Zhang, and Zihan Zhang.
\newblock Sampling proper colorings on line graphs using $(1+o(1))\Delta$
colors.
\newblock In \emph{56th Annual ACM Symposium on Theory of Computing
(STOC)}, pages 1688--1699, 2024.
\newblock \href{https://doi.org/10.1145/3618260.3649724}{doi:10.1145/3618260.3649724}.
\newblock Full version: \href{https://arxiv.org/abs/2307.08080v2}{arXiv:2307.08080v2}.

\end{thebibliography}
\end{document}